\documentclass[10pt,twocolumn]{IEEEtran}
\usepackage{amsmath}
\usepackage{amsthm}
\usepackage{amsfonts}
\usepackage{amssymb}
\usepackage{fullpage}
\usepackage{mathrsfs}
\usepackage{graphicx}
\usepackage{cite}
\newtheorem{definition}{Definition}
\newtheorem{lemma}{Lemma}
\newtheorem{theorem}{Theorem}

\newcommand{\norm}[1]{\left|\left|#1\right|\right|}

\IEEEoverridecommandlockouts
\IEEEpubid{\makebox[\columnwidth]{\bf Draft made on 03/02/2013\hfill} \hspace{\columnsep}\makebox[\columnwidth]{ }}

\title{Improved Bounds on RIP for Generalized Orthogonal Matching Pursuit} 

\author{{Siddhartha Satpathi$^1$, Rajib Lochan Das$^2$ and Mrityunjoy
Chakraborty $^3$ }\\
Department of Electronics and Electrical Communication
Engineering\\
Indian Institute of Technology, Kharagpur, INDIA\\
E.Mail :$^1$sidd.piku@gmail.com, $^2$rajib.das.iit@gmail.com,
 $^3$mrityun@ece.iitkgp.ernet.in}
 
\begin{document}

\maketitle

\begin{abstract} Generalized Orthogonal Matching Pursuit (gOMP) is a natural extension of OMP algorithm where unlike OMP, it may select $N (\geq1)$ atoms
in each iteration. In this paper, we demonstrate that gOMP can successfully reconstruct a $K$-sparse signal from a compressed measurement 
$ {\bf y}={\bf \Phi x}$ by $K^{th}$ iteration if
the sensing matrix ${\bf \Phi}$ satisfies restricted isometry property (RIP) of order $NK$ where $\delta_{NK} < \frac {\sqrt{N}}{\sqrt{K}+2\sqrt{N}}$.
Our bound offers an improvement over the very recent result shown in \cite{wang_2012b}. Moreover, we present another bound for gOMP of order $NK+1$ with
 $\delta_{NK+1} < \frac {\sqrt{N}}{\sqrt{K}+\sqrt{N}}$ which exactly relates to the near optimal bound of $\delta_{K+1} < \frac {1}{\sqrt{K}+1}$ for OMP ($N=1$)
 as shown in \cite{wang_2012a}.

%
\end{abstract}

\section{Introduction} 
Compressed sensing or compressive sampling
(CS)\cite{Donoho,Candes,Baraniuk} is a powerful technique to
represent signals at a sub-Nyquist sampling rate while retaining
the capacity of perfect (or near perfect) reconstruction of the
signal, provided the signal is known to be sparse in some domain.
In last few years, the CS technique has attracted considerable
attention from across a wide array of fields like applied
mathematics, statistics, and engineering, including signal
processing areas such as MR imaging, speech processing, analog to
digital conversion etc. Let a real valued, band-limited signal be sampled following Nyquist
sampling rate over a finite observation interval,
generating a $n\times 1$ signal vector ${\bf u}=( u_1,u_2
,\cdots ,u_n)^T$. The vector ${\bf u}$ is known to be $K$-sparse under
some transform domain
\begin{equation}
 {\bf u}={\bf \Psi}{\bf x} \nonumber
\end{equation}
where ${\bf \Psi}$ is $n\times n$ transform matrix and ${\bf x}$ is the corresponding $n$ dimensional transform coefficient vector that is
 approximated with at most $K$ non-zero entries. Suppose that the signal ${\bf u}$ is
converted  to a lower dimension ($m$) via linear random projection
\begin{equation}
 {\bf y}={\bf A}{\bf u} \nonumber
\end{equation}
where ${\bf y}$ is observation vector with $m\ll n$ and ${\bf A}$ is a flat $m\times n$ random matrix.
According to the CS theory, it is then possible to reconstruct the signal ${\bf u}$ exactly from a very limited number of 
 measurements $M=\mathcal{O}(K \text{log}_e n )$. Therefore, CS framework results in a potential challenge in reconstructing a $K$-sparse signal from
a under determined system equation
\begin{equation}
 {\bf y}={\bf \Phi}{\bf x} \nonumber
\end{equation}
where ${\bf \Phi}={\bf A}{\bf \Psi}$ is a $m\times n$ dimensional sensing matrix.

Under the $K$-sparse assumption ${\bf x}$
can be reconstructed by solving the following $l_0$ minimization
problem
\begin{equation}
 \min_{{\bf x}\in\mathbb{R}^n}\norm{{\bf x}}_0 \textit{ subject to} \hspace{2mm}{\bf y}={\bf \Phi}{\bf
 x}.
\end{equation}
[Note that uniqueness of the $K$-sparse solution requires
every $2K$ column of ${\bf \Phi}$ to be linearly
independent.]
 The above $l_0$ minimization problem provides the sparsest solution for ${\bf x}$.
 However, the $l_0$ minimization problem is a non-convex problem
 and is NP-hard. The feasible practical algorithm for this inverse problem may be broadly classified into two categories, namely convex relaxation 
and  greedy pursuits.
\subsubsection{Convex Relaxation}
This approach translates the non-convex $l_0$ problem into relaxed convex problem using its closest convex $l_1$ norm. This  imposes  ``Restricted Isometry Property (RIP)" condition of appropriate order on ${\bf \Phi}$ as defined below.
\begin{definition}
 A matrix ${\bf\Phi}^{m \times n}$ satisfies RIP of order K if there exists a constant $\delta\in(0,1)$ for all index set $I \subset \{1,2, \cdots ,N\}$ with 
$|I|\leq K$ such that 
\begin{align}
(1-\delta)||{\bf q}||_2^2\leq||{\bf \Phi}_I {\bf q}||_2^2\leq(1+\delta)||{\bf q}||_2^2 .
\end{align}
The RIP constant $\delta_{K}$ is defined as the smallest
value of all $\delta$ for which the RIP is satisfied.
\end{definition}

There are three main directions under this category, namely the
basis pursuit (BP) \cite{BP1}, the basis pursuit de-noising (BPDN)
\cite{Saunders} and the LASSO \cite{lasso}. The reconstruction
problem is formulated under them as,
\begin{eqnarray}
&1.& \textit{BP:}\hspace{2mm} \min_{{\bf x}\in\mathbb{R}^N}\norm{{\bf x}}_1 \textit{ subject to} \hspace{2mm}{\bf y}={\bf \Phi}{\bf x} \nonumber\\
&2.&  \textit{BPDN:}\hspace{2mm} \min_{{\bf x}\in\mathbb{R}^N}\lambda\norm{{\bf x}}_1+\norm{{\bf r}}^2_2\textit{ s.t} \hspace{2mm}{\bf r}={\bf y}-{\bf \Phi}{\bf x}\nonumber\\
&3.&  \textit{LASSO:}\hspace{2mm} \min_{{\bf x}\in\mathbb{R}^N}\norm{{\bf y}-{\bf \Phi}{\bf x}}^2_2\textit{ s.t} \hspace{2mm}\norm{{\bf x}}_1\leqslant\epsilon\nonumber
\end{eqnarray}
The BP problem can be solved by standard polynomial time
algorithms of linear programming. The exact K-sparse signal
reconstruction by BP algorithm based on RIP was first investigated
in \cite{Tao} with the following bound on $\delta$ : $
 \delta_K+\delta_{2K}+\delta_{3K}<1$.
Later the bound was refined as  $\delta_{2K}<\sqrt{2}-1$
\cite{candesRIP}, $\delta_{1.75K}<\sqrt{2}-1$ \cite{Cai}
 and  $\delta_{2K}<0.4652$ \cite{Foucart}.
The BPDN and LASSO problem can be solved by efficient quadratic
programming (QP) like primal-dual interior method. However, the
regularization parameters $\lambda$ and $\epsilon$ play a crucial
role in the performance of these algorithms. 
The convex relaxation technique provides uniform guarantee for
sparse recovery. However, the complexity of $\ell_1$ minimization technique is large enough ($\mathcal{O}(n^3)$) for some applications (e.g. real time video processing).

\subsubsection{Greedy Pursuits}
This approach recovers the $K$-sparse signal by iteratively
constructing the support set  of the sparse signal (\textit{i.e.} index of
non-zero elements in the sparse vector). At each iteration, it
updates its support set by appending the index of one or more
columns (called atoms) of the matrix ${\bf \Phi}$ (often called
dictionary) by some greedy principles. 
This category includes algorithms like orthogonal matching
 pursuit (OMP) \cite{OMP_Tropp}, generalized orthogonal matching
 pursuit (gOMP) \cite{wang_2012b,Liu}, orthogonal least square (OLS) \cite{OLS},compressive sampling matching pursuit
(CoSaMP) \cite{Needell}, subspace pursuit (SP) \cite{Dai} and so on. These algorithms offer very
fast convergence rate with high accuracy in reconstruction performance, but they lack
proper theoretical convergence guaranty. Among these greedy algorithms, OMP is widely used because of its 
simplicity. The theoretical guaranty of OMP algorithm for an exact recovery
of the sparse signal under a $K+1^{th}$ order RIP condition on
${\bf \Phi}$ is improved in the following way:
 $\delta_{K+1}<\frac{1}{3\sqrt{K}}$ in \cite{Davenport}, $\delta_{K+1}<\frac{1}{1+2\sqrt{K}}$ in \cite{Huang},
$\delta_{K+1}<\frac{1}{\sqrt{2K}}$ in \cite{Liu}  and $\delta_{K+1} < \frac {1}{\sqrt{K}+1}$ in \cite{wang_2012a,Qun} .

\subsection{Our contribution in this paper}
In this paper, we have analyzed the theoretical performance of gOMP algorithm in a different approach and our theoretical result improves the bound
on RIP of order $NK$ from $\delta_{NK} < \frac {\sqrt{N}}{\sqrt{K}+3\sqrt{N}}$ \cite{wang_2012b} to $\delta_{NK} < \frac {\sqrt{N}}{\sqrt{K}+2\sqrt{N}}$.
 we have also presented another approach which results in a RIP bound of order $NK+1$ with 
 $\delta_{NK+1} < \frac {\sqrt{N}}{\sqrt{K}+\sqrt{N}}$. Finally, we have discussed the theoretical performance of this algorithm under noisy measurement
and proposed a bound on signal to noise ratio (SNR=$\frac{||{\bf y}||_2}{||{\bf n}||_2}$) for correct reconstruction of support set.

\subsection{Organization of the paper}
Rest of the paper is organized as follows. Next section presents the notations used in this paper and a brief review of OMP and gOMP  
 algorithms. 
In section III, theoretical analysis of gOMP algorithm for noiseless observations is presented. In section IV, analysis of this algorithm in presence of
noise is provided. Discussion is presented in section V and conclusions are drawn in section VI.

\section{Notations and a brief review of OMP and \textit{g}OMP algorithms}

\subsection{Notations} 
The following natations  will be used in this paper. Let the columns of ${\bf \Phi}$ matrix be called as atoms where
${\bf \Phi} = [{\bf \phi}_1 {\bf \phi}_2 {\bf \phi}_3 ... {\bf \phi}_n]$. The matrix ${\bf \Phi}_A$ represents the sub-matrix of ${\bf\Phi}$ with columns indexed by the elements 
present in set A. Similarly {\bf x}$_A$ represents the sub-vector of  {\bf x} with elements whose indices are given in set A. T is the true support set of
 {\bf x} and $\Lambda^k$ is the estimated support set after k iterations of algorithm. 
${\bf \Phi}_{\Lambda^k}^\dag = ({\bf \Phi}_{\Lambda^k}^T{\bf \Phi}_{\Lambda^k})^{-1}{\bf \Phi}_{\Lambda^k}^T$  is the pseudo-inverse of ${\bf \Phi}i_{\Lambda^k}$. 
Here we assume that  ${\bf \Phi}_{\Lambda^k}$ has full column rank ($\Lambda^k<m$). ${\bf P}_{\Lambda^k}={\bf \Phi}_{\Lambda^k}{\bf \Phi}_{\Lambda^k}^\dag$ is 
the projection operator onto column space of  ${\bf \Phi}_{\Lambda^k}$ and ${\bf P}_{\Lambda^k}^\perp={\bf I}-{\bf P}_{\Lambda^k}$ is the projection operator 
upon the rejection space of span($\Lambda^k$). ${\bf A}_{\Lambda^k}={\bf P}_{\Lambda^k}^\perp{\bf \Phi}$ is a matrix obtained by orthogonalizing 
(projecting onto rejection space) the columns of ${\bf \Phi}$ against span(${\bf \Phi}_{\Lambda^k}$).\\
For referring to previous results we use the following notation. Suppose an equation follows from the result of Lemma 1 then L1 
is mentioned at the top of the inequality/equality like $\stackrel{L1}{>}$. Similarly if an equation follows from another equation 
or defination or theorem then it is mentioned as $\stackrel{(1)}{=}$ or $\stackrel{D1}{=}$ or $\stackrel{T1}{>}$ respectively.

\subsection{A brief review of OMP and gOMP algorithms}
The algorithm is presented in Table 1. The OMP algorithm starts with an empty support set ${\Lambda}^0$ and keep selecting a single atom in every iteration 
based on highest correlation with residual signal ${\bf r}^{k-1}$ until the support set is full with the index of $K$ atoms.
At $k^{th}$ iteration, the residual signal ${\bf r}^{k}$ is updated using the difference between signal ${\bf y}$ and its 
orthogonal projection on the subspace spanned by the atoms corresponding to the current support set ${\Lambda}^{k}$.
Generalized OMP algorithm is very similar to OMP where  N largest correlated atoms are selected in each step. This simple modification in identification 
step results in improved reconstruction performance for $K$-sparse signal \cite{wang_2012b}.

\begin{table}
\caption {OMP algorithm}

\begin{tabular}{p{7cm}}
\hline
\textbf{Input}: measurement {\bf y}$\in \mathbb{R}^m$,sensing matrix${\bf \Phi}^{m\times n}$\\
\textbf{Initialization}: counter k=0, residue ${\bf r}^0$={\bf y},\\ estimated support set $\Lambda^k=\emptyset$\\
\hline
\textbf{While}  k$<$K and $|| {\bf r}^k||_2<|| {\bf r}^{k-1}||_2$\\
k=k+1\\
\textit{Identification}: $h^k$=arg max$_j$ $|\langle r^{k-1},{\bf \Phi}_j\rangle|$\\
\textit{Augment}: $\Lambda^k=\Lambda^{k-1} \cup\{h^k\}$\\
\textit{Estimate}: ${\bf x}_{\Lambda^k}=$arg min$_{z:supp(z)\in\Lambda^k}||y-{\bf \Phi}_{\Lambda^k}z||_2$\\
\textit{Update}: ${\bf r}^k=y-{\bf \Phi}_{\Lambda^k}{\bf x}_{\Lambda^k}$\\ 
\textbf{End While}\\
\textbf{Output}: ${\bf x}=$arg min$_{z:supp(z)\in\Lambda^K}||y-{\bf \Phi}_{\Lambda^k}z||_2$\\
\hline In \textbf{gOMP algorithm} the Identification step is \\only different. We select a vector ${\bf h}^k\in \mathbb{R}^N$\\ which has N largest entries in $|{\bf \Phi}^Tr^{k-1}|$.($NK<m$)\\ \hline
\end{tabular}
\end{table}


\section{Analysis of \textit{g}OMP}
To analyse gOMP algorithm we use some commonly used properties  of RIP as summarized in \textit{Lemma 1}.
\begin{lemma}{(Lemma 1 in \cite{Dai}\cite{Needell})}\\
\hspace*{3mm} a) $\delta_{K_1}<\delta_{K_2} \forall K_1<K_2$\text{ (monotonicity)}\\ 
\vspace{2mm} 
\hspace{2mm} b) $(1-\delta_{|I|})||{\bf q}||_2\leq||{\bf \Phi}_I^T {\bf \Phi}_I{\bf q}||_2\leq(1+\delta_{|I|})||{\bf q}||_2$\\
 \vspace{2mm} 
\hspace{2mm} c) $||{\bf \Phi}_I^T{\bf q}||_2<\sqrt{1+\delta_{|I|}}||{\bf q}||_2$\\ 
 \vspace{2mm}
 d) $\langle{\bf \Phi}_I {\bf q},{\bf \Phi}_J {\bf p}\rangle\leq||{\bf q}||_2||{\bf \Phi}_I^T{\bf \Phi}_J{\bf p}||_2<\delta_{|I|+|J|}||{\bf p}||_2 ||{\bf q}||_2$\\
for $I,J\subset \{1,2,...,n\},|I|$, $ {\bf q}\in\mathbb{R}^I$ and $ {\bf p}\in\mathbb{R}^J$
\end{lemma}
Note that, the algorithm can reconstruct a $K$-sparse signal by $K^{th}$ iterations if atleast one correct index is chosen in each iteration.
 Now, let in $k+1^{th}$ iteration $\beta^k_i=\langle{\bf \Phi},{\bf r^k}\rangle$ for 
$i\in T$ and $\alpha^k_j=\langle{\bf \Phi}_j,{\bf r^k}\rangle$ for $j\notin T$ where $\beta^k_i$ 's and $\alpha^k_j$ 's 
are arranged in descending order. So $\beta^k_1>\beta^k_2 ...>\beta^k_N$ are N largest correlations in support set and similarly
  $\alpha^k_1>\alpha^k_2 ...>\alpha^k_N$ are N largest correlations of incorrect indices. Now if we ensure that $\beta^k_1>\alpha^k_N$
 then atleast $\beta^k_1$ will appear in the overall N largest correlated atoms which are selected. Hence, we find the lower bound of $\beta^k_1$ 
and upper bound of $\alpha^k_N$ and compare them. In this paper, we propose two RIP bounds which are presented as Theorem 1 and Theorem 2.

\begin{theorem}
gOMP can recover {\bf x} exactly when ${\bf \Phi}$ satisfies RIP of order $NK$ with 
\begin{equation}\delta_{NK} < \frac {\sqrt{N}}{\sqrt{K}+2\sqrt{N}}\nonumber
\end{equation}
\end{theorem}


\begin{proof} To start with we use the same upper bound on $\alpha^k_N$ as presented in \cite{wang_2012b}. Interested readers may refer \cite{wang_2012b} for proof of the following lemma.

\begin{lemma}{(Lemma 3.6 in \cite{wang_2012b})} \\
$\alpha^k_N< \dfrac{\delta_{NK}}{1-\delta_{NK}} \dfrac{||{\bf x}_{T-\Lambda^k}||_2}{\sqrt{N}}$
\end{lemma}

Now we go about finding a better bound on $\beta^k_1$. We observe that ${\bf r}^k={\bf P}_{\Lambda^k}^\perp {\bf y}$ and $\beta^k_1=||{\bf \Phi}^T_T{\bf r}^k||_\infty$ and for any $i\in\Lambda^k$
\begin{align}\langle{\bf \phi}_i,{\bf r}^k\rangle= \langle{\bf \phi}_i,{\bf P}^\perp_{\Lambda^k}{\bf r}^k\rangle=\langle {\bf P}^\perp_{\Lambda^k}{\bf \phi}_i , {\bf r}^k\rangle=0
\end{align}
So,
\begin{align}
\nonumber&||{\bf \Phi}_T^T {\bf r}^k||_\infty 
>\frac{1}{\sqrt{K}}||{\bf \Phi}^T_T {\bf r}^k||_2 \text{ (as } |T|=K) \\
\nonumber&\stackrel{(3)}{=}\frac{1}{\sqrt{K}}||{\bf \Phi}_{T-\Lambda^k}^T {\bf r}^k||_2=\frac{1}{\sqrt{K}}||{\bf \Phi}^T_{T-\Lambda^k} {\bf P}_{\Lambda^k}^\perp {\bf y}||_2 \\
\nonumber&=\frac{1}{\sqrt{K}}||{\bf \Phi}^T_{T-\Lambda^k} ({\bf P}_{\Lambda^k}^\perp)^T {\bf P}_{\Lambda^k}^\perp {\bf y}||_2 \text{ (as } {\bf P}={\bf P}^T \text{\&} {\bf P}={\bf P}^2)\\
\nonumber&=\frac{1}{\sqrt{K}}||({\bf P}_{\Lambda^k}^\perp {\bf \Phi}_{T-\Lambda^k} )^T {\bf P}_{\Lambda^k}^\perp {\bf \Phi}_T {\bf x}_T||_2\\
&\stackrel{(3)}{=}\frac{1}{\sqrt{K}}||({\bf P}_{\Lambda^k}^\perp {\bf \Phi}_{T-\Lambda^k} )^T {\bf P}_{\Lambda^k}^\perp {\bf \Phi}_{T-\Lambda^k} {\bf x}_{T-\Lambda^k}||_2
\end{align} 
Now, to proceed further  we require the following lemma.

\begin{lemma}{(Extension of lemma 3.2 from \cite{Davenport})}\\
 ${\bf A}_{I_1}^{m\times n}$ satisfies modified RIP of $\frac{\delta_{|I_1|+|I_2|}}{1-\delta_{|I_1|+|I_2|}}$ $$ (1-\frac{\delta_{|I_1|+|I_2|}}{1-\delta_{|I_1|+|I_2|}})||{\bf u}||_2^2<||{\bf A}_{I_1}{\bf u}||_2^2<(1+\delta_{|I_1|+|I_2|})||{\bf u}||_2^2$$
and also
$$ (1-(\frac{\delta_{|I_1|+|I_2|}}{1-\delta_{|I_1|+|I_2|}})^2)||{\bf \Phi}{\bf u}||_2^2<||{\bf A}_{I_1}{\bf u}||_2^2<||{\bf \Phi}{\bf u}||_2^2$$
where {\bf u}$\in\mathbb{R}^n$, $I_1,I_2\in\{1,...,n\}$ $supp({\bf u})\in I_2$ and $I_1\cap I_2=\emptyset$
\end{lemma}

 \textit{Proof:} In Appendix A

Now Let ${\bf x}'=\begin{bmatrix}{\bf x}_{T-\Lambda^k} \\  0\end{bmatrix}$ ,  ${\bf x}'\in\mathbb{R}^n$. 
So, ${\bf A}_{\Lambda^k}{\bf x}'={\bf P}^\perp_{\Lambda^k}{\bf \Phi} {\bf x}'={\bf P}^\perp_{\Lambda^k}\begin{bmatrix}{\bf \Phi}_{T-\Lambda^k} &  {\bf \Phi}_{(T-\Lambda^k)^c}\end{bmatrix}\begin{bmatrix}{\bf x}_{T-\Lambda^k} \\ 0\end{bmatrix}={\bf P}^\perp_{\Lambda^k}{\bf \Phi}_{T-\Lambda^k}{\bf x}_{T-\Lambda^k}$
and also ${\bf P}^\perp_{\Lambda^k}{\bf \Phi}_{T}{\bf x}_T={\bf P}^\perp_{\Lambda^k}\begin{bmatrix}{\bf \Phi}_{T-\Lambda^k} &  {\bf \Phi}_{(T\cap\Lambda^k)}\end{bmatrix}\begin{bmatrix}{\bf x}_{T-\Lambda^k} \\  {\bf x}_{T\cap\Lambda^k}\end{bmatrix} \stackrel{(3)}{=} {\bf P}^\perp_{\Lambda^k}{\bf \Phi}_{T-\Lambda^k}{\bf x}_{T-\Lambda^k}$. Hence we get
 \begin{align}\nonumber&{\bf A}_{\Lambda^k} {\bf x}'={\bf P}^\perp_{\Lambda^k}{\bf \Phi} {\bf x}'={\bf P}^\perp_{\Lambda^k}{\bf \Phi}_T {\bf x}_T\\&={\bf r}^k=
{\bf P}^\perp_{\Lambda^k}{\bf \Phi}_{T-\Lambda^k} {\bf x}_{T-\Lambda^k}
\end{align}
Hence, from Lemma 3, with $I_1=\Lambda^k$ and $I_2=supp({\bf x}')=T-\Lambda^k$  and $I_1\cap I_2=\emptyset ,|I_1|+|I_2|=Nk+K-l$ where $T\cap\Lambda^k=l$ we get 
\begin{align}
\nonumber&||{\bf A}_{\Lambda^k}{\bf x}'||_2^2 > (1-\dfrac{\delta_{Nk+K-l}}{1-\delta_{Nk+K-l}})||{\bf x}'||_2^2\\
&\stackrel{L1a}{>}(1-\dfrac{\delta_{NK}}{1-\delta_{NK}})||{\bf x}_{T-\Lambda^k}||_2^2
\end{align} 
Moreover,
\begin{align}
\nonumber&||{\bf A}_{\Lambda^k}{\bf x}'||_2^2 \stackrel{(5)}{=}||{\bf P}^\perp_{\Lambda^k}{\bf \Phi}_{T-\Lambda^k} {\bf x}_{T-\Lambda^k}||_2^2\\
\nonumber&=\langle {\bf P}^\perp_{\Lambda^k}{\bf \Phi}_{T-\Lambda^k} {\bf x}_{T-\Lambda^k},{\bf P}^\perp_{\Lambda^k}{\bf \Phi}_{T-\Lambda^k} {\bf x}_{T-\Lambda^k}\rangle\\
\nonumber&=\langle({\bf P}^\perp_{\Lambda^k}{\bf \Phi}_{T-\Lambda^k})^TP^\perp_{\Lambda^k}{\bf \Phi}_{T-\Lambda^k} {\bf x}_{T-\Lambda^k}, {\bf x}_{T-\Lambda^k}\rangle\\
&<||({\bf P}^\perp_{\Lambda^k}{\bf \Phi}_{T-\Lambda^k})^TP^\perp_{\Lambda^k}{\bf \Phi}_{T-\Lambda^k} {\bf x}_{T-\Lambda^k} ||_2 ||  {\bf x}_{T-\Lambda^k}||_2
\end{align}
Combining (6) and (7) we get 
\begin{align}\nonumber&||({\bf P}^\perp_{\Lambda^k}{\bf \Phi}_{T-\Lambda^k})^TP^\perp_{\Lambda^k}{\bf \Phi}_{T-\Lambda^k} {\bf x}_{T-\Lambda^k} ||_2\\
&>  (1-\dfrac{\delta_{NK}}{1-\delta_{NK}})||{\bf x}_{T-\Lambda^k}||_2
\end{align}
Therefore combining this result with (4) and (8) we get \begin{equation}\beta^k_1>\dfrac{1}{\sqrt{K}}(1-\dfrac{\delta_{NK}}{1-\delta_{NK}})||{\bf x}_{T-\Lambda^k}||_2\end{equation}
Making lower bound on $\beta^k_1$ greater than upper bound of $\alpha^k_N$ (from Lemma 2) bring us to the result.
\end{proof}

The next theorem states the other bound for gOMP success.

\begin{theorem}
 gOMP can recover {\bf x} exactly when ${\bf \Phi}$ satisfies RIP of order $NK+1$ with 
\begin{equation}\delta_{NK+1} < \frac {\sqrt{N}}{\sqrt{K}+\sqrt{N}}\nonumber
\end{equation}
\end{theorem}

\begin{proof} Let us begin by examining the residue ${\bf r}^k$ in the $k+1^{th}$ iteration. In \cite{wang_2012a} it was shown that ${\bf r}^k\in span({\bf \Phi}_T)$ for OMP where estimated support set $\Lambda^k \subset T$. Now we show that in cases where $\Lambda^k$ and T are in general modelled as shown in Fig.1, ${\bf r}^k$ is indeed spanned by ${\bf \Phi}_{T\cup\Lambda^k}$.
\begin{align}
{\bf r}^k\nonumber&={\bf y}-{\bf \Phi}_{\Lambda^k}{\bf \Phi}^\dag_{\Lambda^k}{\bf y}\\
\nonumber&={\bf \Phi}_T {\bf x}_T-{\bf P}_{\Lambda^k}{\bf y}\\
\nonumber&={\bf \Phi}_{T-\Lambda^k} {\bf x}_{T-\Lambda^k}+{\bf \Phi}_{T\cap\Lambda^k} {\bf x}_{T\cap\Lambda^k}\\
\nonumber&\hspace{4mm}-{\bf P}_{\Lambda^k}({\bf \Phi}_{T-\Lambda^k} {\bf x}_{T-\Lambda^k}+{\bf \Phi}_{T\cap\Lambda^k} {\bf x}_{T\cap\Lambda^k})\\
\nonumber&={\bf \Phi}_{T-\Lambda^k} {\bf x}_{T-\Lambda^k}-{\bf P}_{\Lambda^k}{\bf \Phi}_{T-\Lambda^k} {\bf x}_{T-\Lambda^k}\\
&={\bf \Phi}_{T-\Lambda^k} {\bf x}_{T-\Lambda^k}-{\bf \Phi}_{\Lambda^k}{\bf z}_{\Lambda^k}\\
&={\bf \Phi}_{T\cup \Lambda^k}{\bf x}''_{T \cup \Lambda^k}
\end{align}

\begin{figure} [ht]
\centering
\includegraphics[width=50mm,height=30mm]{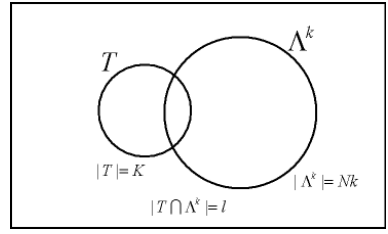}
\caption{Venn diagram for correct and estimated support set}
\label{fig1}
\end{figure}

where (10) follows from the fact that ${\bf P}_{\Lambda^k}{\bf \Phi}_{T-\Lambda^k} {\bf x}_{T-\Lambda^k}\in span({\bf \Phi}_{\Lambda^k})$ and it
 can be viewed as ${\bf \Phi}_{\Lambda^k}{\bf z}_{\Lambda^k}$ where ${\bf z}_{\Lambda^k}\in\mathbb{R}^{\Lambda^k} $ and ${\bf z}_{\Lambda^k}={\bf \Phi}^\dag_{\Lambda^k}{\bf \Phi}_{T-\Lambda^k} {\bf x}_{T-\Lambda^k}$. So ${\bf x}''_{T \cup \Lambda^k}$
 is a vector in $\mathbb{R}^{T\cup\Lambda^k}$. Observe that 
\begin{align}{\bf x}''_{T \cup \Lambda^k}=\begin{bmatrix}x_{T-\Lambda^k}\\z_{\Lambda^k}\end{bmatrix}
\end{align}

Let W be the set of remaining incorrect indices over which $\alpha^k_i $'s are chosen ($W\in (T\cup\Lambda^k)^c$). So,
\begin{align} \alpha^k_N \nonumber&=min(\langle\Phi_i,{\bf r}^k\rangle)\hspace{3mm} (i\in W)\\
\nonumber& < \frac{\sum\alpha_i}{N}<\sqrt{\frac{\sum\alpha_i^2}{N}} \text{ (as }|W|=N)\\
\nonumber&=\frac{||{\bf \Phi}_W^T{\bf r}^k||_2}{\sqrt{N}}\\
\nonumber&\stackrel{(11)}{=}\frac{1}{\sqrt{N}}||{\bf \Phi}_W^T{\bf \Phi}_{T\cup\Lambda^k}{\bf x}''_{T\cup\Lambda^k}||_2\\
\nonumber&\stackrel{L1d}{=}\frac{1}{\sqrt{N}}\delta_{N+Nk+K-l}||{\bf x}''_{T\cup\Lambda^k}||_2\\
&\stackrel{L1a}{<}\frac{1}{\sqrt{N}}\delta_{NK+1}||{\bf x}''_{T\cup\Lambda^k}||_2
\end{align}
where (13) comes from the fact that $l\geq k$ and $k\leq K-1$.

Now for finding lower bound of $\beta^k_1$ in terms of $||{\bf x}''_{T\cup\Lambda^k}||_2$ we proceed in this way.
\begin{align}
\nonumber&\beta^k_1=||{\bf \Phi}_T^T {\bf r}^k||_\infty \\
\nonumber& >\frac{1}{\sqrt{K}}||{\bf \Phi}^T_T {\bf r}^k||_2 \hspace{4mm}    (as |T|=K) \\
& =\frac{1}{\sqrt{K}}||\begin {bmatrix}{\bf \Phi}_T & {\bf \Phi}_{\Lambda^k-T}\end{bmatrix}^T {\bf r}^k||_2  \\
\nonumber&=\frac{1}{\sqrt{K}}||{\bf \Phi}^T_{T\cup\Lambda^k}{\bf \Phi}_{T\cup\Lambda^k}{\bf x}''_{T\cup\Lambda^k}||_2 \\
\nonumber&\stackrel{L1b}{>}\frac{1}{\sqrt{K}}(1-\delta_{Nk+K-l})||{\bf x}''_{T\cup\Lambda^k}||_2 \\
&\stackrel{L1a}{>}\frac{1}{\sqrt{K}}(1-\delta_{NK})||{\bf x}''_{T\cup\Lambda^k}||_2
\end{align}
where (14) comes as ${\bf \Phi}^T_{\Lambda^k-T} {\bf r}^k=0$ which follows from (3). Now from (13) and (15) ensuring  $\beta^k_1>\alpha^k_N$ gives us the result.
\end{proof}

\section{Analysis in presence of noise}
In case of noise we can model the measurement as ${\bf y'}={\bf y}+{\bf n}={\bf \Phi} {\bf x}+{\bf n}$ where ${\bf n}$ is the added noise. We can show the performance of this algorithm in presence of noise in two ways. One is by finding the upper bound of reconstruction error energy $||{\bf x}-{\bf x}_{\Lambda^k}||_2$ (as presented in \cite{wang_2012b})and other is by providing a condition for exact reconstruction subject to upper bound on measurement SNR=$\frac{||{\bf y}||_2}{||{\bf n}||_2}$.

\begin{theorem}
 If $k=K$ forms the stopping criterion in gOMP with ${\bf \Phi}$ satisfying $\delta_{NK+K}<1$ and $\delta_{NK}<\frac {\sqrt{N}}{\sqrt{K}+2\sqrt{N}}$ then $||{\bf x}-{\bf x}_{\Lambda^k}||_2<C_{K_1}||{\bf n}||_2$ where 

$C_{K_1}=\frac{(1-\delta_{NK})(\sqrt{N}(1+\delta_K)+\sqrt{K(1+\delta_N)(1+\delta_K)})}{(\sqrt{N}-(\sqrt{K}+2\sqrt{N})\delta_{NK})\sqrt{1-\delta_{NK+K}}} + \frac{2}{\sqrt{1-\delta_{NK+K}}}$

\end{theorem}

\textit{Proof:} In Appendix B

\begin{theorem}
 If $k=K$ forms the stopping criterion in gOMP with ${\bf \Phi}$ satisfying $\delta_{NK+K}<1$ and $\delta_{NK+1}<\frac {\sqrt{N}}{\sqrt{K}+\sqrt{N}}$ then $||{\bf x}-{\bf x}_{\Lambda^k}||_2<C_{K_2}||{\bf n}||_2$ where 

$C_{K_2}=\frac{(\sqrt{N}(1+\delta_K)+\sqrt{K(1+\delta_N)(1+\delta_K)})}{(\sqrt{N}-(\sqrt{K}+\sqrt{N})\delta_{NK+1})\sqrt{1-\delta_{NK+K}}} + \frac{2}{\sqrt{1-\delta_{NK+K}}}$

\end{theorem}

\textit{Proof:} In Appendix C

The above bounds provide a estimate on upper bound on reconstruction energy. But they do not guarantee estimation of correct support set. In some cases it may happen that the reconstruction error energy is bounded but the chosen support set is completely different. This may prove expensive because in most cases knowledge of correct support set is more important than knowledge of exact value at that position. So easily verifiable bounds guaranteeing reconstruction of correct support set are necessary.

In communication we often judge the performance by the SNR of the received signal. Before applying sparse reconstruction we do not have the information of energy of vector {\bf x}. But we do have knowledge of energy of clean measurement from transmitter's end. Hence by calculating the SNR at the receiver's end we can have an idea whether a particular algorithm can be implemented for reconstruction or not. This can be a good measure of performance analysis for reconstruction algorithms. So we present a bound on $\frac{||y||_2}{||n||_2}$ for which correct support set is estimated. 
Before stating the theorem let us analyse the assumption made: $|{\bf x}_i|>\frac{|{\bf x}_j|}{\gamma} \forall i,j\in T$. This implies that all non zero values of {\bf x} are bounded within some ratio of the maximum. We see that the sparse systems are modelled by setting the values of elements in {\bf x} below some threshold as zero. Hence this assumption can always be made. If suppose {\bf x} has a non-zero value below $\frac{|x_{max}|}{\gamma}$ then it can be modelled as a $K-1$ sparse system by setting that value to zero without affecting the output much.

\begin{theorem}
 If measurement ${\bf y}={\bf \Phi} {\bf x}$ is corrupted with noise ${\bf n}$ then gOMP algorithm can still recover the true support of {\bf x} provided $\frac{||{\bf y}||_2}{||{\bf n}||_2}>C_{K_3}$ where $$C_{K_3}= \frac{\sqrt{K(1+\delta_N)}+\sqrt{N(1+\delta_K)}}
{\frac{\sqrt{N}(1-\delta_K)(1-2\delta_{NK})}{\gamma\sqrt{(1+\delta_K)}(1-\delta_{NK})^2}-\sqrt{K(1+\delta_N)}}$$
\end{theorem}

\begin{proof} At first let us make use of the assumption and provide a result which would be used in subsequent proof

\begin{lemma}
 With usual notations we see that
 $$||{\bf \Phi}_{T-\Lambda^k}{\bf x}_{T-\Lambda^k}||_2>\sqrt{\frac{{(K-l)(1-\delta_K)}}{{K(1+\delta_K)}}}\frac{||{\bf \Phi}_Tx_T||_2}{\gamma}$$
\end{lemma}
\textit{proof:} In Appendix D
 
 Let us again compute the bounds on $\alpha^k_N$ and $\beta^k_1$
\begin{align}
\nonumber&\alpha^k_N<\frac{||{\bf \Phi}_W^T {\bf P}_{\Lambda^k}^\perp {\bf y'}||_2}{\sqrt{N}}\\
\nonumber&\stackrel{L1c}{<} \frac{\sqrt{1+\delta_N}||{\bf P}_{\Lambda^k}^\perp {\bf y'}||_2}{\sqrt{N}}\\
\nonumber&< \frac{\sqrt{1+\delta_N}||{\bf y'}||_2}{\sqrt{N}}\\
&< \frac{\sqrt{1+\delta_N}(||{\bf \Phi}_T{\bf x_T}||_2+||{\bf n}||_2)}{\sqrt{N}}
\end{align}

\begin{align}
\nonumber&\beta^k_1>\frac{1}{\sqrt{K-l}}||{\bf \Phi}^T_{T-\Lambda^k} {\bf P}_{\Lambda^k}^\perp ({\bf \Phi}_{T-\Lambda^k} {\bf x}_{T-\Lambda^k}+{\bf n})||_2\\
\nonumber&\stackrel{L1c}{>}\frac{1}{\sqrt{K-l}}||{\bf \Phi}^T_{T-\Lambda^k} {\bf P}_{\Lambda^k}^\perp {\bf \Phi}_{T-\Lambda^k} {\bf x}_{T-\Lambda^k}||_2\\
\nonumber&-\frac{\sqrt{1+\delta_K}}{\sqrt{K}}||{\bf n}||_2\\
\nonumber&\stackrel{(7)}{>}\frac{||{\bf A}_{\Lambda^k} {\bf x}'||_2^2}{\sqrt{K-l}||{\bf x}_{T-\Lambda^k}||_2}-\frac{\sqrt{1+\delta_K}}{\sqrt{K}}||{\bf n }||_2\\
\nonumber&\stackrel{L3}{>}\frac{(1-(\frac{\delta_{NK}}{1-\delta_{NK}})^2)||{\bf \Phi}_{T-\Lambda^k}{\bf x}_{T-\Lambda^k}||^2_2}{\sqrt{K-l}||{\bf x}_{T-\Lambda^k}||_2}\\
\nonumber&-\frac{\sqrt{1+\delta_K}}{\sqrt{K}}||{\bf n}||_2\\
\nonumber&\stackrel{D1}{>}\frac{(1-2\delta_{NK})\sqrt{1-\delta_K}||{\bf \Phi}_{T-\Lambda^k}{\bf x}_{T-\Lambda^k}||_2}{\sqrt{K-l}(1-\delta_{NK})^2}\\
\nonumber&-\frac{\sqrt{1+\delta_K}}{\sqrt{K}}||{\bf n}||_2\\
&\stackrel{L4}{>}\frac{(1-2\delta_{NK})(1-\delta_K)||{\bf \Phi}_T{\bf x}_T||_2}{(1-\delta_{NK})^2\gamma\sqrt{K(1+\delta_K)}}-\frac{\sqrt{1+\delta_K}}{\sqrt{K}}||{\bf n}||_2
\end{align}

Making $\alpha^k_N<\beta^k_1$ ( (16) $<$ (17) ) for correct choice of index we get the desired SNR bound.
\end{proof}

\section{Discussion}The proposed bound  in Theorem 1 is better than the one from \cite{wang_2012b} because while obtaining lower bound of $\beta^k_1$ instead of applying successive inequalities we use a more direct inequality presented in Lemma 3 which leads us to a higher lower bound. This bound is also better than the bound proposed in \cite{Liu} ($\delta_{NK}<\frac{\sqrt{K}}{(2+\sqrt{2})\sqrt{K}}$) for $N<1.45K$. But according to \cite{wang_2012b} gOMP performs better than OMP for small values of N only.

 It is difficult to compare the bounds presented in Theorem 1 and Theorem 2 since $\delta_{NK}<\delta_{NK+1}$ and $\frac {\sqrt{N}}{\sqrt{K}+\sqrt{N}}>\frac {\sqrt{N}}{\sqrt{K}+2\sqrt{N}}$. But intuitively we can see that bound on $\delta_{NK+1}$ is more optimal since it reduces to near optimal bound on OMP for special case of $N=1$. The proposition on SNR seems to be a good approach since it is an easily measurable quantity and can be used in future research for comparing greedy algorithm's performance under noise.

 \section{Conclusions}
In this paper, we have given an elegent proof of the theoretical performance of  gOMP algorithm. Our analysis  improves the bound on 
RIP  of order $NK$ from $\delta_{NK} < \frac {\sqrt{N}}{\sqrt{K}+3\sqrt{N}}$ \cite{wang_2012b} to $\delta_{NK} < \frac {\sqrt{N}}{\sqrt{K}+2\sqrt{N}}$.
In the same paper, we have presented another bound of order $NK+1$ with RIP constant
 $\delta_{NK+1} < \frac {\sqrt{N}}{\sqrt{K}+\sqrt{N}}$. We have also presented improved theoretical performance of gOMP algorithm under noisy measurements.

\appendices

\renewcommand{\theequation}{A.\arabic{equation}}
\numberwithin{equation}{section}
\section{Proof of lemma 3}
We know that
\begin{align}\nonumber&||{\bf \Phi} {\bf u}||_2^2=||{\bf P}_{I_1}{\bf \Phi} {\bf u}||_2^2+||{\bf P}_{I_1}^\perp {\bf \Phi} {\bf u}||_2^2 \\
\Rightarrow &||{\bf A}_{I_1}{\bf u}||_2^2=||{\bf \Phi} {\bf u}||_2^2-||{\bf P}_{I_1}{\bf \Phi} {\bf u}||_2^2
\end{align}
Now $\langle {\bf P}_{I_1}{\bf \Phi} {\bf u},{\bf \Phi}{\bf u}\rangle=({\bf \Phi} {\bf u})^T {\bf P}_{I_1}^T {\bf \Phi} {\bf u}=({\bf \Phi} {\bf u})^T {\bf P}_{I_1}^T{\bf P}_{I_1} {\bf \Phi} {\bf u}=||{\bf P}_{I_1}{\bf \Phi} {\bf u}||_2^2$. Further
we see that ${\bf P}_{I_1}{\bf \Phi} {\bf u}\in span({\bf \Phi}_{I_1})$. So ${\bf P}_{I_1}{\bf \Phi} {\bf u}={\bf \Phi} {\bf z}$ for some ${\bf z}\in\mathbb{R}^n$  with $supp({\bf z})\in{I_1}$.

\begin{align}
\nonumber&\frac{||{\bf P}_{I_1}{\bf \Phi} {\bf u}||_2}{||{\bf \Phi}{\bf u}||_2}
=\frac{\langle {\bf P}_{I_1}{\bf \Phi}{\bf u},{\bf \Phi} {\bf u}\rangle}{||{\bf P}_{I_1}{\bf \Phi} {\bf u}||_2||{\bf \Phi}{\bf u}||_2}=\frac{\langle {\bf \Phi} {\bf z},{\bf \Phi} {\bf u}\rangle} {||{\bf \Phi} {\bf z}||_2||{\bf \Phi} {\bf u}||_2}\\
\nonumber&\stackrel{D1, L1d}{<}\frac{\delta_{|I_1|+|I_2|}}{\sqrt{1-\delta_{|I_1|}}\sqrt{1-\delta_{|I_2|}}}\\
&\stackrel{L1a}{<}\frac{\delta_{|I_1|+|I_2|}}{1-\delta_{|I_1|+|I_2|}}
\end{align} 
So from (A.1) and (A.2) we get$ ||{\bf A}_{I_1}{\bf u}||_2^2>(1-(\dfrac{\delta_{|I_1|+|I_2|}}{1-\delta_{|I_1|+|I_2|}})^2)||{\bf \Phi} {\bf u}||_2^2>(1-(\dfrac{\delta_{|I_1|+|I_2|}}{1-\delta_{|I_1|+|I_2|}})^2)(1-\delta_{|I_2|})||{\bf u}||_2^2\stackrel{L1a}{>}(1-\dfrac{\delta_{|I_1|+|I_2|}}{1-\delta_{|I_1|+|I_2|}})||{\bf u}||_2^2$.
Applying $||{\bf P}_{I_1}{\bf \Phi} {\bf u}||_2^2>0$ in (A.1) the upper bound becomes$||{\bf A}_{I_1}{\bf u}||_2^2<||{\bf \Phi} {\bf u}||_2^2<(1+\delta_{|I_2|})||{\bf u}||_2^2\stackrel{L1a}{<}(1+\delta_{|I_1|+|I_2|})||{\bf u}||_2$.

\section{Proof of theorem 3}First we need to find the bounds on $\alpha^k_N$ and $\beta^k_1$ in presence of noise in k+1$^{th}$  iteration
\begin{align}
\nonumber&\alpha^k_N<\frac{||{\bf \Phi}_W^T{\bf r}^k||_2}{\sqrt{N}}\\
\nonumber&=\frac{1}{\sqrt{N}}||{\bf \Phi}_W^T({\bf P}^\perp_{\Lambda^k}{\bf y}+{\bf P}^\perp_{\Lambda^k}{\bf n})||_2\\
\nonumber&<\frac{1}{\sqrt{N}}(||{\bf \Phi}_W^TP^\perp_{\Lambda^k}{\bf y}||_2+||{\bf \Phi}_W^TP^\perp_{\Lambda^k}{\bf n}||_2)\\
\nonumber&\stackrel{L2,L1c}{<}\frac{1}{\sqrt{N}}\frac{\delta_{NK}}{1-\delta_{NK}}||{\bf x}_{T-\Lambda^k}||_2+\frac{\sqrt{1+\delta_N}}{\sqrt{N}}||{\bf P}^\perp_{\Lambda^k}{\bf n}||_2\\
&<\frac{1}{\sqrt{N}}\frac{\delta_{NK}}{1-\delta_{NK}}||{\bf x}_{T-\Lambda^k}||_2+\frac{\sqrt{1+\delta_N}}{\sqrt{N}}||{\bf n}||_2
\end{align}
and
\begin{align}
\nonumber&\beta^k_1>\frac{1}{\sqrt{K}}||{\bf \Phi}^T_{T-\Lambda^k} {\bf P}_{\Lambda^k}^\perp {\bf y'}||_2\\
\nonumber&>\frac{1}{\sqrt{K}}||{\bf \Phi}^T_{T-\Lambda^k} {\bf P}_{\Lambda^k}^\perp {\bf \Phi}_T {\bf x}_T||_2-\frac{1}{\sqrt{K}}||{\bf \Phi}^T_{T-\Lambda^k} {\bf P}_{\Lambda^k}^\perp {\bf n}||_2\\
\nonumber&\stackrel{(9),L1c}{>}\frac{1}{\sqrt{K}}(1-\frac{\delta_{NK}}{1-\delta_{NK}})||{\bf x}_{T-\Lambda^k}||_2-\frac{\sqrt{1+\delta_K}}{\sqrt{K}} ||{\bf P}_{\Lambda^k}^\perp {\bf n}||_2\\
&>\frac{1}{\sqrt{K}}(1-\frac{\delta_{NK}}{1-\delta_{NK}})||{\bf x}_{T-\Lambda^k}||_2-\frac{\sqrt{1+\delta_K}}{\sqrt{K}} ||{\bf n}||_2
\end{align}
Now at the end of algorithm it may happen that some incorrect atoms are chosen. Lets say this happens for the first time in the p+1$^{th}$ step. Then at this particular step (B.1)$>$(B.2). Which implies

\begin{align}
||{\bf x}_{T-\Lambda^p}||_2<\nonumber&\frac{(1-\delta_{NK})(\sqrt{N(1+\delta_K)}+\sqrt{K(1+\delta_N)})}{(\sqrt{N}-(\sqrt{K}+2\sqrt{N})\delta_{NK})} \\ &\times||{\bf n}||_2
\end{align}

The error in reconstruction energy can be seen as 
\begin{align}
\nonumber&||{\bf x}-{\bf x}_{\Lambda^K}||_2\stackrel{D1}{<}\frac{||{\bf \Phi} {\bf x}-{\bf \Phi} {\bf x}_{\Lambda^K}||_2}{\sqrt{1-\delta_{NK+K}}}\\
\nonumber&=\frac{||{\bf \Phi} {\bf x}-{\bf \Phi} {\bf \Phi}^\dag_{\Lambda^K}{\bf y'}||_2}{\sqrt{1-\delta_{NK+K}}}\\
\nonumber&=\frac{||{\bf y'}-{\bf P}_{\Lambda^K}{\bf y'}-{\bf n}||_2}{\sqrt{1-\delta_{NK+K}}}\\
\nonumber&<\frac{||{\bf P}^\perp_{\Lambda^K}{\bf y'}||_2+||{\bf n}||_2}{\sqrt{1-\delta_{NK+K}}}\\
\nonumber&=\frac{||{\bf r}^K||_2+||{\bf n}||_2}{\sqrt{1-\delta_{NK+K}}}\\
\nonumber&\leq\frac{||{\bf r}^p||_2+||{\bf n}||_2}{\sqrt{1-\delta_{NK+K}}}(\text{as }||{\bf r}^i||_2\leq||{\bf r}^j||_2  \text{ for } i>j)\\
\nonumber&=\frac{||{\bf P}^\perp_{\Lambda^p}({\bf \Phi}_T {\bf x}_T+{\bf n})||_2+||{\bf n}||_2}{\sqrt{1-\delta_{NK+K}}}\\
\nonumber&<\frac{||{\bf P}^\perp_{\Lambda^p}{\bf \Phi}_T {\bf x}_T||_2+||{\bf P}^\perp_{\Lambda^p}{\bf n}||_2+||{\bf n}||_2}{\sqrt{1-\delta_{NK+K}}}\\
\nonumber&<\frac{||{\bf P}^\perp_{\Lambda^p}{\bf \Phi}_{T-\Lambda^p} {\bf x}_{T-\Lambda^p}||_2+2||{\bf n}||_2}{\sqrt{1-\delta_{NK+K}}}\\
\nonumber&<\frac{||{\bf \Phi}_{T-\Lambda^p} {\bf x}_{T-\Lambda^p}||_2+2||{\bf n}||_2}{\sqrt{1-\delta_{NK+K}}}\\
&\stackrel{D1}{<}\frac{\sqrt{1+\delta_{K}}||{\bf x}_{T-\Lambda^p}||_2+2||{\bf n}||_2}{\sqrt{1-\delta_{NK+K}}}
\end{align}
By using (B.3) in (B.4) we get the desired bound.

\section{Proof of  theorem 4}In this case we find the bounds on $\alpha^k_N$ and $\beta^k_1$ in presence of noise similar to our proof on second bound of gOMP. Proceding similar to (B.1) and (B.2) we get
\begin{align}
\alpha^k_N<\frac{1}{\sqrt{N}}\delta_{NK+1}||{\bf x}''_{T\cup\Lambda^k}||_2+\frac{\sqrt{1+\delta_N}}{\sqrt{N}}||{\bf n}||_2
\end{align}
\begin{align}
\beta^k_1>\frac{1-\delta_{NK}}{\sqrt{K}}||{\bf x}''_{T\cup\Lambda^k}||_2-\frac{\sqrt{1+\delta_K}}{\sqrt{K}} ||{\bf n}||_2
\end{align}
So failure at p+1$^{th}$ step implies
\begin{equation}
||{\bf x}''_{T\cup\Lambda^p}||_2<\frac{\sqrt{N(1+\delta_K)}+\sqrt{K(1+\delta_N)}}{(\sqrt{N}-(\sqrt{K}+\sqrt{N})\delta_{NK+1})} ||n||_2
\end{equation}
Now to get an upper bound in estimation error we proceed similarly as in (B.4)
\begin{align}
\nonumber&||{\bf x}-{\bf x}_{\Lambda^K}||_2<\frac{\sqrt{1+\delta_{K}}||{\bf x}_{T-\Lambda^p}||_2+2||{\bf n}||_2}{\sqrt{1-\delta_{NK+K}}}\\
\nonumber&\stackrel{(12)}{<}\frac{\sqrt{1+\delta_{K}}||{\bf x}''_{T\cup\Lambda^p}||_2+2||{\bf n}||_2}{\sqrt{1-\delta_{NK+K}}}\\
&\stackrel{(C.3)}{<}C_{K_2}||{\bf n}||_2
\end{align}

\section{Proof of Lemma 4}

\begin{align}
\nonumber&||{\bf \Phi}_{T-\Lambda^k} {\bf x}_{T-\Lambda^k}||_2\stackrel{D1}{>}\sqrt{1-\delta_K}||{\bf x}_{T-\Lambda^k}||_2\\
\nonumber&>\sqrt{(1-\delta_K)(K-l)}|{\bf x}_{min}|\\
\nonumber&>\sqrt{\frac{(1-\delta_K)(K-l)}{K}}\frac{||{\bf x}_T||_2}{\gamma}\\
&\stackrel{D1}{>}\sqrt{\frac{(1-\delta_K)(K-l)}{K(1+\delta_K)}}\frac{||{\bf \Phi}_Tx_T||_2}{\gamma}
\end{align}



\bibliographystyle{IEEEtran}
\bibliography{refs}
\end{document}